\documentclass[aps,pra,showpacs,twocolumn]{revtex4}
\usepackage{amsmath, amssymb}
\usepackage{graphicx}
\usepackage{subfigure}
\usepackage{xcolor}
\usepackage{amsthm}

\newtheorem{lemma}{Lemma}
\newtheorem{thrm}{Theorem}

\renewcommand{\vec}[1]{\mathbf{#1}}

\DeclareMathOperator{\re}{Re}
\DeclareMathOperator{\im}{Im}

\begin{document}

\title{Entanglement conditions and polynomial identities}

\author{E. Shchukin}
\email{evgeny.shchukin@gmail.com}
\affiliation{Lehrstuhl f\"ur Modellierung und Simulation,
Universit\"at Rostock, D-18051 Rostock, Germany}

\begin{abstract}
We develop a rather general approach to entanglement characterization based on
convexity properties and polynomial identities. This approach is applied to
obtain simple and efficient entanglement conditions which work equally well in
both discrete as well as continuous-variable environments. Examples of violations
of our conditions are presented.
\end{abstract}


\maketitle

One of the most fascinating features of Quantum Mechanics is the phenomenon of
entanglement. Despite its long history, originated with the seminal Einstein,
Podolsky and Rosen paper \cite{PhysRev.47.777} in 1935, many questions about
this phenomenon are still open. The basic problem is the ability to distinguish
entangled states from non-entangled (i.e. separable) ones. The first complete
characterization of entangled states has been given in \cite{PhysRevA.69.022308}
in the form of a hierarchy of conditions whose first level is the famous
Peres-Horodecki positivity of partial transposition condition
\cite{PhysRevLett.77.1413, Horodecki19961} (another approach has been developed in
\cite{PhysRevA.79.022318}). This hierarchy is a numerical
procedure, with all restrictions and limitations of a numerical algorithm in a
setting where an analytical description of the problem at hand is desirable.
That is why the existence of a complete numerical characterization does not
eliminate the need in analytical conditions.

A real-valued function $f(\hat{\varrho})$, defined on the set of all density
operators, is called convex if the inequality 
\begin{equation}\label{eq:convex} 
    f\Bigl(\sum_n p_n \hat{\varrho}_n\Bigr) \leqslant \sum_n p_n f(\hat{\varrho}_n) 
\end{equation} 
is valid for all states $\hat{\varrho}_n$ and for all numbers $p_n \geqslant 0$
with the normalization condition $\sum_n p_n = 1$. The simplest example of 
a convex function is the average value $\langle\hat{A}\rangle_{\hat{\varrho}}$
for a fixed operator $\hat{A}$, for which the condition \eqref{eq:convex} becomes 
equality. If $F: \mathbf{R} \to \mathbf{R}$ is a convex function of a real argument
then $F(\langle\hat{A}\rangle_{\hat{\varrho}})$ is also convex. For example,
$|\langle\hat{A}\rangle_{\hat{\varrho}}|$ and $\langle\hat{A}\rangle^2_{\hat{\varrho}}$ 
are convex functions of $\hat{\varrho}$ (for fixed $\hat{A}$). In fact,
$\langle\hat{A}\rangle^n_{\hat{\varrho}}$ is convex if $n = 2n'$ is an even number.

A real-valued function $g(\hat{\varrho})$ is called concave if the inequality,
opposite to the one given by Eq.~\eqref{eq:convex}, is valid
\begin{equation}\label{eq:concave}
    g\Bigl(\sum_n p_n \hat{\varrho}_n\Bigr) \geqslant \sum_n p_n g(\hat{\varrho}_n).
\end{equation}
In other words, a function is concave if it increases on mixtures of states.
Examples of a concave function are the average
$\langle\hat{A}\rangle_{\hat{\varrho}}$ and the variation
$\sigma^2_{\hat{A}}(\hat{\varrho}) = \langle\hat{A}^2\rangle -
\langle\hat{A}\rangle^2$ of the observable $\hat{A}$. Note, that the average
$\langle\hat{A}\rangle_{\hat{\varrho}}$, being a linear function of the quantum
state, is both convex and concave. A trivial special case is the operator $\hat{A} = \alpha \cdot \hat{1}$ of multiplication
by a constant $\alpha$, or simply a constant operator, for which $\langle\hat{A}\rangle_{\hat{\varrho}} = \alpha$ is both convex and
concave.

After these preliminary definitions we can formulate the following lemma, which
plays a central role in our work:
\begin{lemma}\label{lemma1}
Let $f(\hat{\varrho})$ be
a convex function and $g(\hat{\varrho})$ be concave. If the inequality 
\begin{equation}\label{eq:fg}
    f(\hat{\varrho}_n) \leqslant g(\hat{\varrho}_n) 
\end{equation}
is valid for some states $\hat{\varrho}_n$, $n = 1, 2, \ldots$, then it is also
valid for all their mixtures $\sum_n p_n \hat{\varrho}_n$, where $p_n \geqslant
0$ and $\sum_n p_n = 1$.
\end{lemma}
\begin{proof}
The proof is rather simple: for a given mixture $\hat{\varrho} = \sum_n p_n \hat{\varrho}_n$ we have the inequalities
\begin{equation}
    f(\hat{\varrho}) \leqslant \sum_n p_n f(\hat{\varrho}_n) \leqslant \sum_n p_n g(\hat{\varrho}_n)
    \leqslant g(\hat{\varrho}).
\end{equation}
Roughly speaking, the left-hand side does not increase on mixtures, but the right-hand side does not decrease, which completes
the proof.
\end{proof}

As a consequence of this lemma, we obtain the following result. If the
inequality \eqref{eq:fg} is valid for all factorizable states $\hat{\varrho}_n =
\hat{\varrho}^a_n \otimes \hat{\varrho}^b_n \otimes \ldots \otimes
\hat{\varrho}^z_n$ then it is also valid for all separable states. A violation
of this inequality is a sufficient condition for entanglement. This
characterization is complete in the sense that for any entangled quantum state
$\hat{\varrho}$ one can find such a pair of functions $f$ and $g$ which satisfy
Lemma~\ref{lemma1} (with $\hat{\varrho}_n$ being separable states) and for which
$f(\hat{\varrho}) > g(\hat{\varrho})$. In fact, as it has been proven in
\cite{Horodecki19961}, for any entangled state $\hat{\varrho}$ there is an
observable $\hat{E}$ such that $\langle\hat{E}\rangle_{\hat{\tau}} \geqslant 0$
for all separable states $\hat{\tau}$, but for which
$\langle\hat{E}\rangle_{\hat{\varrho}} < 0$. Let us take $f(\hat{\varrho}) = \langle 0 \cdot \hat{1} \rangle_{\hat{\varrho}} = 0$
and $g(\hat{\varrho}) = \langle\hat{E}\rangle_{\hat{\varrho}}$, then $f$ is
convex (and concave) and $g$ is concave (and convex). According to the choice of $\hat{E}$ the inequality
\eqref{eq:fg} is satisfied for all separable states, but $f(\hat{\varrho}) >
g(\hat{\varrho})$ for the given entangled state $\hat{\varrho}$.

Though this characterization is general and complete, it is not very useful in
practice since it is not clear how to choose appropriate functions $f$ and $g$.
Thereby, the main difficulty to get efficient entanglement conditions using 
Lemma~\ref{lemma1} is to find pairs of functions $f$ and $g$, one of which is
convex and the other is concave, and which satisfy the inequality \eqref{eq:fg}
for all factorizable states. Moreover, it must be possible to violate this
inequality. Finding such pairs of function is the main result of our work.

Before we formulate our first result, we prove a couple of simple lemmas. The
following one
deals with the variances of factorizable states.
\begin{lemma}\label{lemma-1}
Factorizable quantum states $\hat{\varrho} = \hat{\varrho}^a \otimes
\hat{\varrho}^b$ satisfy the inequality $\sigma_{AB}(\hat{\varrho}^a \otimes
\hat{\varrho}^b) \geqslant \sigma_A(\hat{\varrho}^a) \sigma_B(\hat{\varrho}^b)$
for arbitrary observables $\hat{A}$ and $\hat{B}$ of different degrees of
freedom.
\end{lemma}
\begin{proof}
In fact, for the factorizable state $\hat{\varrho}$ we have the relation
$\langle\hat{A}\hat{B}\rangle_{\hat{\varrho}} =
\langle\hat{A}\rangle_{\hat{\varrho}^a}\langle\hat{B}\rangle_{\hat{\varrho}^b}$
and analogously for the squares of observables, so we have
\begin{equation}
    \sigma^2_{AB}(\hat{\varrho}) = \langle\hat{A}^2\rangle_{\hat{\varrho}^a}\langle\hat{B}^2\rangle_{\hat{\varrho}^b} -
    \langle\hat{A}\rangle^2_{\hat{\varrho}^a}\langle\hat{B}\rangle^2_{\hat{\varrho}^b}.
\end{equation}
Using the identities of the form $\langle\hat{A}^2\rangle =
\sigma^2_A(\hat{\varrho}) + \langle\hat{A}\rangle^2$, we obtain the relation
$\sigma^2_{AB}(\hat{\varrho}) = \sigma^2_A(\hat{\varrho}^a)\sigma^2_B(\hat{\varrho}^b) + \ldots$,
where the dots on the right-hand side of this equality stand for the sum
$\sigma^2_A(\hat{\varrho}^a)\langle\hat{B}\rangle^2_{\hat{\varrho}^b}+
\langle\hat{A}\rangle^2_{\hat{\varrho}^a}\sigma^2_B(\hat{\varrho}^b)$. Since
this expression is always non-negative we get the desired inequality.
\end{proof}

The next lemma concerns quadratic polynomials, non-negative on the right half of
the real line and puts some condition on the coefficients of such polynomials.
\begin{lemma}\label{lemma-abc}
If the inequality $a\lambda^2 - b\lambda + c \geqslant 0$,
is valid for all $\lambda \geqslant 0$, where $a, b, c \geqslant 0$, then $b^2 \leqslant 4ac$.
\end{lemma}
\begin{proof}
In fact, if $a = 0$ then $b$ must be zero too, in other case this inequality, 
which in this case reads as $-b\lambda + c \geqslant 0$, is
violated for sufficiently large $\lambda$. If $a > 0$ then the left-hand side of
the inequality can be written as
$a\lambda^2 - b\lambda +c = a(\lambda - b/(2a))^2 + c - b^2/(4a)$,
and for $\lambda = b/(2a) \geqslant 0$ it must be non-negative, from which we
obtain the desired inequality $b^2 \leqslant 4ac$.
\end{proof}

Now we are ready to prove one of our main results.
\begin{thrm}\label{thm1}
For arbitrary observables $\hat{A}$, $\hat{A}'$ and $\hat{B}$, $\hat{B}'$ of different degrees of freedom
the inequality
\begin{equation}\label{eq:main1}
    \sigma_{AB} \sigma_{A'B'} \geqslant \frac{1}{4} |\langle[\hat{A},\hat{A}'][\hat{B},\hat{B}']\rangle|
\end{equation}
is valid for all bipartite separable quantum states.
\end{thrm}
\begin{proof}
Using Lemma \ref{lemma-1}, we can estimate the quantity $\sigma^2_{AB}(\hat{\varrho}) + \sigma^2_{A'B'}(\hat{\varrho})$ for a factorizable state
$\hat{\varrho} = \hat{\varrho}^a \otimes \hat{\varrho}^b$ as follows:
$\sigma^2_{AB} + \sigma^2_{A'B'} \geqslant
\sigma^2_A(\hat{\varrho}^a) \sigma^2_B(\hat{\varrho}^b) +\sigma^2_{A'}(\hat{\varrho}^a) \sigma^2_{B'}(\hat{\varrho}^b)$. 
The right-hand side of this inequality is bounded from below by the product $2
\sigma_A(\hat{\varrho}^a) \sigma_{A'}(\hat{\varrho}^a) \sigma_B(\hat{\varrho}^b)
\sigma_{B'}(\hat{\varrho}^b)$. Applying the Heisenberg uncertainty relation for
the observables $\hat{A}$ and $\hat{A}'$ to the product of the first two variances
(expressed as
$\sigma_{\hat{A}}(\hat{\varrho}^a)\sigma_{\hat{A}'}(\hat{\varrho}^a) \geqslant
(1/2)|\langle[\hat{A}, \hat{A}']\rangle|_{\hat{\varrho}^a}$) and the same
relation for the observables $\hat{B}$ and $\hat{B}'$ to the product of the last two
variances and taking into account the identity $\langle[\hat{A},
\hat{A}']\rangle_{\hat{\varrho}^a} \langle[\hat{B},
\hat{B}']\rangle_{\hat{\varrho}^b} = \langle[\hat{A}, \hat{A}'][\hat{B},
\hat{B}']\rangle_{\hat{\varrho}}$, we finally get the inequality
\begin{equation}\label{eq:sigma}
    \sigma^2_{AB}(\hat{\varrho}) + \sigma^2_{A'B'}(\hat{\varrho}) \geqslant \frac{1}{2} |\langle[\hat{A}, \hat{A}'][\hat{B}, \hat{B}']\rangle_{\hat{\varrho}}|,
\end{equation}
for factorizable states. As we have already mentioned, the variance
$\sigma^2_{AB}(\hat{\varrho})$ is a concave function of $\hat{\varrho}$ (for
fixed observables $\hat{A}$ and $\hat{B}$) and the quantity
$|\langle\hat{C}\rangle_{\hat{\varrho}}|$ is convex, so we can apply
Lemma~\ref{lemma1} and conclude that the inequality \eqref{eq:sigma} is valid
for all bipartite separable states. Now lets scale the observables $\hat{A}$ and
$\hat{B}$ by a non-negative factor $\sqrt{\lambda}$, $\lambda \geqslant 0$:
$\hat{A} \to \sqrt{\lambda} \hat{A}$, $\hat{B} \to \sqrt{\lambda} \hat{B}$.
After scaling we get that any bipartite separable quantum state satisfies the
inequality
$\sigma^2_{AB} \lambda^2 - (1/2)|\langle[\hat{A}, \hat{A}'][\hat{B}, \hat{B}']\rangle|\lambda + \sigma^2_{A'B'} \geqslant 0$
for all non-negative $\lambda$. Applying Lemma \ref{lemma-abc}, we get the desired inequality \eqref{eq:main1}.
\end{proof}

The proof of this theorem shows that the inequality given by
Eq.~\eqref{eq:main1} immediately follows from Lemma~\ref{lemma-1} for a
factorizable state. But one cannot use Lemma~\ref{lemma-1} for a general
separable state, since the left-hand side of the inequality \eqref{eq:main1},
which is a product of two convex functions, may not be convex, as the product
$x^2(x-1)^2$ of two functions of real argument illustrates. Both factors
(parabolas) are convex but their product is not. That is why we need the trick
with the scaling of observables.

It is interesting to compare the inequality \eqref{eq:main1} and the Heisenberg
uncertainty relation for the observables $\hat{A}\hat{B}$ and $\hat{A}'\hat{B}'$.
The latter gives us the following lower bound for the product of the dispersions
of the observables under study: $\sigma_{AB} \sigma_{A'B'} \geqslant
(1/2)|\langle[\hat{A}\hat{B}, \hat{A}'\hat{B}']\rangle|$. Note that this
inequality is valid for all bipartite quantum states. Our inequality
\eqref{eq:main1} gives another lower bound for the same product of dispersions,
but this bound is guaranteed to be valid only for separable states (there can be
entangled states which also satisfy this bound). If, for a given state and
observables, we have
$|\langle[\hat{A},\hat{A}'][\hat{B},\hat{B}']\rangle| > 2|\langle[\hat{A}\hat{B}, \hat{A}'\hat{B}']\rangle|$,
then the inequality \eqref{eq:main1} gives a stronger restriction then the
Heisenberg relation so it may be possible to violate it (this means
to calculate the product of dispersions and compare it with 
$(1/4)|\langle[\hat{A},\hat{A}'][\hat{B},\hat{B}']\rangle|$). In other case, the
Heisenberg uncertainty relation prohibits violations of the inequality
\eqref{eq:main1}. As we will see below, our inequality can be violated in many
important cases, both discrete as well as continuous variable ones. As a measure 
of violation we use the quantity
\begin{equation}
    V = \frac{|\langle[\hat{A},\hat{A}'][\hat{B},\hat{B}']\rangle|}{4\sigma_{AB} \sigma_{A'B'}}.
\end{equation}
According to the inequality \eqref{eq:main1}, $V \leqslant 1$ for all separable states, so
$V > 1$ is a sufficient condition for entanglement.

Our first example of violation of the inequality \eqref{eq:main1} is for the
especially simple choice of the observables $\hat{A} = \hat{x}_a$, $\hat{A}' =
\hat{p}_a$, $\hat{B} = \hat{p}_b$ and $\hat{B}' = \hat{x}_b$, the position and momentum operators. In this case the
inequality \eqref{eq:main1} reads as
\begin{equation}\label{eq:xp}
    \sigma_{xp} \sigma_{px} \geqslant \frac{1}{4}.
\end{equation}
We have $|\langle[\hat{A}\hat{B}, \hat{A}'\hat{B}']\rangle| =
|\im\langle\hat{a}^2-\hat{b}^2\rangle|$, where $\hat{a}$
and $\hat{b}$ are annihilation operators of the first and the second mode respectively, and in many cases this quantity is zero.
Let us take the state 
\begin{equation}\label{eq:c}
    |\psi\rangle = \sum^{+\infty}_{n=0}c_n|2n,2n\rangle
\end{equation}
with real coefficients $c_n$, where $|n\rangle$ are the Fock states. One can easily check that for this state we have 
$\langle\hat{x}_a \hat{p}_b\rangle = \langle\hat{p}_a \hat{x}_b\rangle = 0$.
To compute the dispersions note that $\langle\hat{x}^2_a \hat{p}^2_b\rangle 
= \langle\hat{p}^2_a \hat{x}^2_b\rangle$ and that
\begin{equation}
    \hat{x}^2_a \hat{p}^2_b = \frac{1}{4}\bigl((2\hat{a}^\dagger\hat{a}+1)(2\hat{b}^\dagger\hat{b} + 1)
    - \hat{a}^2\hat{b}^2 - \hat{a}^{\dagger 2}\hat{b}^{\dagger 2}\bigr) +\ldots,
\end{equation}
where the dots stand for the terms that do not contribute to $\langle\hat{x}^2_a \hat{p}^2_b\rangle$.
The product of the dispersions reads as
\begin{equation}
    \sigma_{xp} \sigma_{px} = \frac{1}{4} + Q(c_0, c_1, \ldots),
\end{equation}
where $Q = Q(c_0, c_1, \ldots)$ is a quadratic form of the coefficients $c_n$, $n=0, 1, \ldots$
\begin{equation}
    Q = \sum^{+\infty}_{n=0} \bigl(2n(2n+1)c^2_n - (n+1)(2n+1)c_nc_{n+1}\bigr).
\end{equation}
The matrix of this quadratic form is given by
\begin{equation}
    C = 
    \begin{pmatrix}
        0 & 1/2 & 0 & 0 & 0 & \ldots \\
        1/2 & 6 & 3 & 0 & 0 & \ldots \\
        0 & 3 & 20 & 15/2 & 0 & \ldots \\
        0 & 0 & 15/2 & 42 & 14 & \ldots \\
        0 & 0 & 0 & 14 & 72 & \ldots \\
        \hdotsfor{6}
    \end{pmatrix}.
\end{equation}
Let us consider a truncated $(N+1) \times (N+1)$ matrix $C_{N}$, which corresponds 
to the truncated state \eqref{eq:c} (where $c_n = 0$ if $n > N$). The minimal
value of the expression $\vec{c}^T C_N \vec{c}$ for $\vec{c} = (c_0, c_1,
\ldots, c_N)$ with the normalization condition  $|\vec{c}| = 1$ is the minimal
eigenvalue $\lambda_{\mathrm{min}}$ of the matrix $C_N$ and the corresponding
eigenvector gives the coefficients $c_k$, $k = 0, 1, \ldots, N$. With \textsl{Mathematica} 
we have obtained that
$\lambda_{\mathrm{min}}$ tends to the value $\approx -0.04495$ as $n \to
+\infty$. This corresponds to the maximal value 
\begin{equation}
    V_{\mathrm{max}} = \frac{1/4}{(1/4)+\lambda_{\mathrm{min}}} \approx 1.2192.
\end{equation}
For the truncated state $|\psi_2\rangle = c_0|00\rangle + c_1|22\rangle$, where
$c_0$ and $c_1$ are real and $c^2_0 + c^2_1 = 1$, we have $\sigma^2_{xp} = \sigma^2_{px} 
= (1/4)+6c^2_1 - c_0
c_1$. If we take the coefficients such that $0 < c_1 < c_0/6$ then we get
$\sigma_{xp} \sigma_{px} < 1/4$. The violation $V$ as a function of $c_0$ 
(assuming that $c_1 = \sqrt{1-c^2_0}$) takes the
maximum for $c_0 \approx 0.997$, $c_1 \approx 0.077$ and the maximal value is
$V \approx 1.197$. This means that even the simplest truncated state $|\psi_2\rangle$
gives a rather good violation.

Now consider a mixture of the state \eqref{eq:c} with vacuum
\begin{equation}\label{eq:mixed}
    \hat{\varrho} = p |\psi\rangle\langle\psi| + (1-p)|00\rangle\langle 00|, 
\end{equation}
where $0 \leqslant p \leqslant 1$. It is easy to check that for this mixed state we have 
\begin{equation}
    \sigma_{xp} \sigma_{px} = \frac{1}{4} + pQ(c_0, c_1, \ldots).
\end{equation}
For $p = 0$, when the state \eqref{eq:mixed} is simply vacuum, there is no violation, as it must be. 
But for any $0 < p \leqslant 1$ we have
\begin{equation}
    V_{\mathrm{max}}(p) = \frac{1/4}{1/4 + p \lambda_{\mathrm{min}}} > 1,
\end{equation}
since $\lambda_{\mathrm{min}} < 0$. We see that even a tiny fraction of the entangled state \eqref{eq:c} in the mixture
\eqref{eq:mixed} can be detected with the condition \eqref{eq:main1}.

Let us choose two orthogonal quantum states $|0\rangle$ and $|1\rangle$, $\langle 0|1\rangle = 0$. 
Define the Pauli-like operators $\hat{s}_x$, $\hat{s}_y$, $\hat{s}_z$ and $\hat{s}_0$ as follows:
\begin{equation}\label{eq:s}
\begin{split}
    \hat{s}_x &= |0\rangle\langle 1| + |1\rangle\langle 0|, \quad \hat{s}_y = -i(|0\rangle\langle 1| - |1\rangle\langle 0|), \\
    \hat{s}_z &= |0\rangle\langle 0| - |1\rangle\langle 1|, \quad \hat{s}_0 = |0\rangle\langle 0| + |1\rangle\langle 1|.
\end{split}
\end{equation}
If the system is two-dimensional, then 
$|0\rangle$ and $|1\rangle$ form a basis in the state space of the system and $\hat{s}_0$ is the identity operator $\hat{s}_0 = \hat{1}$. The operators 
\eqref{eq:s} satisfy the well-known spin-like commutation relations
$[\hat{s}_x, \hat{s}_y] = 2i\hat{s}_z$, $[\hat{s}_y, \hat{s}_z] = 2i\hat{s}_x$, $[\hat{s}_z, \hat{s}_x] = 2i\hat{s}_y$,
the equalities $\hat{s}^2_x = \hat{s}^2_y = \hat{s}^2_z = \hat{s}_0$ and the anticommutator relations $\{\hat{s}_x, \hat{s}_y\} = \{\hat{s}_x, \hat{s}_z\} = 
\{\hat{s}_y, \hat{s}_z\}= 0$, where $\{\hat{A}, \hat{A}'\} = \hat{A}\hat{A}' + \hat{A}'\hat{A}$.

We show that our condition can detect entanglement of any pure bipartite state of the form 
\begin{equation}\label{eq:Psi}
    |\Psi\rangle = \alpha |\psi_0\rangle|\varphi_0\rangle + \beta|\psi_1\rangle|\varphi_1\rangle,
\end{equation}
where $|\varphi_0\rangle$ and $|\varphi_1\rangle$ is a pair of orthogonal states of the other mode of the system and 
$|\alpha|^2 + |\beta|^2 = 1$. 
Since \eqref{eq:Psi} is, in fact, the Schimdt decomposition
of the state $|\Psi\rangle$, this state is separable if and only if either $\alpha = 0$ or $\beta = 0$.

One can verify that $\langle\hat{s}^a_x \hat{s}^b_x\rangle = -\langle\hat{s}^a_y \hat{s}^b_y\rangle = 2\re(\alpha\beta^*)$ and 
$\langle\hat{s}^a_x \hat{s}^b_y\rangle = \langle\hat{s}^a_y \hat{s}^b_x\rangle = -2\im(\alpha\beta^*)$, 
$\langle\hat{s}^a_z \hat{s}^b_z\rangle = \langle\hat{s}^a_0 \hat{s}^b_0\rangle = 1$.
Now consider the operators 
\begin{equation}
\begin{split}
    \hat{A} &= \hat{s}^a_x \cos\theta + \hat{s}^a_y \sin\theta, \ \hat{A}' = \hat{s}^a_x \cos\theta' + \hat{s}^a_y \cos\theta', \\
    \hat{B} &= \hat{s}^b_x \cos\eta + \hat{s}^b_y \sin\eta, \ \hat{B}' = \hat{s}^b_x \cos\eta' + \hat{s}^b_y \sin\eta'.
\end{split}
\end{equation}
Then we have
\begin{equation}
\begin{split}
    &\langle\hat{A}\hat{B}\rangle = 2\re(\alpha\beta^* e^{i(\theta+\eta)}), \\
    &\langle\hat{A}'\hat{B}'\rangle = 2\re(\alpha\beta^* e^{i(\theta'+\eta')}),
\end{split}
\end{equation}
$[\hat{A}, \hat{A}'] = 2i\sin(\theta'-\theta)\hat{s}^a_z$ and $[\hat{B}, \hat{B}'] = 2i\sin(\eta'-\eta)\hat{s}^b_z$. 
Let us take the angles
$\theta = -\arg\alpha$, $\eta = \arg\beta$, $\theta' = \theta + \pi/2$, $\eta' = \eta + \pi/2$, 
and we have $\langle\hat{A}\hat{B}\rangle = -\langle\hat{A}'\hat{B}'\rangle = 2|\alpha\beta|$, 
$[\hat{A}, \hat{A}'] = 2i\hat{s}^a_z$ and $[\hat{B}, \hat{B}'] = 2i\hat{s}^b_z$. 
Then the inequality \eqref{eq:main1} reads as
$1 - 4|\alpha\beta|^2 \geqslant 1$,
which is obviously violated unless $\alpha = 0$ or $\beta = 0$, that is unless
the state \eqref{eq:Psi} is factorizable. Note that in this way one can verify
entanglement of any pure state with the Schmidt number larger than 1.

The inequality \eqref{eq:main1} can also be violated by the squeezed vacuum state defined as 
\begin{equation}
    |S\rangle = \sqrt{1-\lambda^2}\sum^{+\infty}_{n=0}\lambda^n |n,n\rangle
\end{equation}
with real
parameter $\lambda$, $|\lambda|<1$. If for both modes we take the operators $\hat{A} = \sum^{+\infty}_{n=0} |2n\rangle\langle 2n+1| + |2n+1\rangle\langle 2n|$ and
$\hat{B} = \sum^{+\infty}_{n=0} -i(|2n\rangle\langle 2n+1| + |2n+1\rangle\langle 2n|)$, we get
\begin{equation}
    V = \left(\frac{1+\lambda^2}{1-\lambda^2}\right)^2.
\end{equation}
We see that in this case we have $V > 1$ for $\lambda \not= 0$ and $V \to
+\infty$ when $\lambda \to 1$.

Let us consider the inequality obtained in \cite{PhysRevLett.99.210405}, which in the bipartite case reads as
\begin{equation}\label{eq:bell}
\begin{split}
    \langle\hat{A}\hat{B}-\hat{A}'\hat{B}'\rangle^2 &+ \langle\hat{A}\hat{B}'+\hat{A}'\hat{B}\rangle^2 \\
    &\leqslant \langle(\hat{A}^2 + \hat{A}^{\prime 2})(\hat{B}^2 + \hat{B}^{\prime 2})\rangle.
\end{split}
\end{equation}
It has been proven in \cite{PhysRevLett.101.040404} that this inequality
cannot be violated in the case of quadratures. For two level systems and dichotomic observables this inequality 
becomes
\begin{equation}
    \langle\hat{A}\hat{B}-\hat{A}'\hat{B}'\rangle^2 + \langle\hat{A}\hat{B}'+\hat{A}'\hat{B}\rangle^2 \leqslant 4,
\end{equation}
which has been proven to be valid for all bipartite two-level quantum states \cite{PhysRevLett.88.230406}. 
In \cite{PhysRevLett.99.210405} it has been shown that the multipartite analogue of the inequality \eqref{eq:bell} can be violated for a
large number of modes (more than $10$). The inequality \eqref{eq:bell} easily follows from our approach with the help of the identity
\begin{equation}\label{eq:abcd}
    (a b - a' b')^2 + (a b' + a' b)^2 = (a^2 + a^{\prime 2})(b^2 + b^{\prime 2}),
\end{equation}
which expresses the multiplicativity of the norm of the complex numbers (consider the product of two complex numbers $a+ia'$ and $b+ib'$). 
Now for any factorizable state we have 
\begin{equation}
\begin{split}
    &\langle\hat{A}\hat{B}-\hat{A}'\hat{B}'\rangle^2 + \langle\hat{A}\hat{B}'+\hat{A}'\hat{B}\rangle^2 
    =\langle\hat{A}\hat{B}\rangle^2 + \langle\hat{A}\hat{B}'\rangle^2 \\
    &+ \langle\hat{A}'\hat{B}\rangle^2 
    + \langle\hat{A}'\hat{B}'\rangle^2 \leqslant \langle(\hat{A}^2 + \hat{A}^{\prime 2})(\hat{B}^2 + \hat{B}^{\prime 2})\rangle.
\end{split}
\end{equation}
As it has already been noted, the left-hand side of this inequality is convex and the right-hand side is concave 
and Lemma~\ref{lemma1} shows that the inequality \eqref{eq:bell} is valid for all separable states. The multipartite analogue of the inequality \eqref{eq:bell} 
can be obtained using the identity for product of more then two complex numbers. Complex numbers is not the only numerical system
whose norm is multiplicative. Inequalities based on the multiplicativity 
of the norms of quaternions and octonions have been obtained in \cite{PhysRevA.78.032104}.

Nothing prevents us from using the identity \eqref{eq:abcd} in the other direction and derive the inequality
\begin{equation}
\begin{split}
    &\langle(\hat{A}\hat{B}-\hat{A}'\hat{B}')^2\rangle + \langle(\hat{A}\hat{B}'+\hat{A}'\hat{B})^2\rangle \\
    &\geqslant \langle\hat{A}\hat{B}\rangle^2 + \langle\hat{A}\hat{B}'\rangle^2 + 
    \langle\hat{A}'\hat{B}\rangle^2 + \langle\hat{A}'\hat{B}'\rangle^2,
\end{split}
\end{equation}
which after some algebraic manipulations can be transformed to the following form:
\begin{equation}
    \sigma^2_{AB} + \sigma^2_{AB'} + \sigma^2_{A'B} + \sigma^2_{A'B'} \geqslant
    |\langle[\hat{A},\hat{A}'][\hat{B},\hat{B}']\rangle|.
\end{equation}
It is clear that this inequality follows from \eqref{eq:main1}.

The theorem \ref{thm1} can be straightforwardly generalized to the multipartite case.
\begin{thrm}
For arbitrary observables $\hat{A}_1$, $\hat{A}'_1$, \ldots, $\hat{A}_n$, $\hat{A}'_n$ of different degrees of freedom the inequality
\begin{equation}
    \sigma_{A_1\ldots A_n} \sigma_{A'_1\ldots A'_n} \geqslant \frac{1}{2^n} |\langle[\hat{A}_1,\hat{A}'_1]\ldots[\hat{A}_n,\hat{A}'_n]\rangle|
\end{equation}
is valid for all completely separable multipartite quantum states.
\end{thrm}
It is easy to see that the Bell state 
\begin{equation}
    |B_n\rangle = \frac{1}{\sqrt{2}}(|0\ldots0\rangle + |1\ldots1\rangle)
\end{equation}
violates this inequality for $\hat{A}_k = \hat{s}^{(k)}_x$,
$\hat{A}'_k = \hat{s}^{(k)}_y$, which in this case reads as
\begin{equation}\label{eq:uvw2}
    \sigma_{x\ldots x}\sigma_{y\ldots y} \geqslant |\langle\hat{s}^{(1)}_z \ldots \hat{s}^{(n)}_z\rangle|.
\end{equation}
If $n$ is even then the state $|B_n\rangle$ is an eigenstate of the operators $\hat{s}^{(1)}_x \ldots \hat{s}^{(n)}_x$, 
$\hat{s}^{(1)}_y \ldots \hat{s}^{(n)}_y$ and $\hat{s}^{(1)}_z \ldots \hat{s}^{(n)}_z$ 
with eigenvalues $\pm 1$, so the left-hand side of the inequality \eqref{eq:uvw2} is zero and the right-hand side is $1$.

Our next result is based on the Ramajuan identities.
\begin{thrm}
For any observables $\hat{A}$, $\hat{A}'$ and $\hat{B}$, $\hat{B}'$ of different degrees of freedom the inequality
\begin{equation}\label{eq:main2}
\begin{split}
    &\langle\hat{A}\hat{B} + \hat{A}\hat{B}' + \hat{A}'\hat{B}\rangle^n +
    \langle\hat{A}\hat{B}' + \hat{A}'\hat{B} + \hat{A}'\hat{B}'\rangle^n \\
    &+\langle\hat{A}\hat{B} - \hat{A}'\hat{B}'\rangle^n \leqslant
    \langle(\hat{A}\hat{B}' - \hat{A}'\hat{B})^n\rangle +\\
    &\langle(\hat{A}'\hat{B} + \hat{A}'\hat{B}' + \hat{A}\hat{B})^n\rangle +
    \langle(\hat{A}'\hat{B}' + \hat{A}\hat{B} + \hat{A}\hat{B}')^n\rangle
\end{split}
\end{equation}
with $n = 2$ and $n = 4$ is valid for all bipartite separable states.
\end{thrm}

The proof is based on the Ramanujan identity
\begin{equation}
\begin{split}
    (a b + a b' &+ a' b)^n + (a b' + a' b + a' b')^n \\
    &+ (a b - a' b')^n = (a' b + a' b' + a b)^n \\
    &+ (a' b' + a b + a b')^n + (a b' - a' b)^n,
\end{split}
\end{equation}
which is valid for $n=2, 4$, see \cite{ramanujan}, and proceeds as above. In this case we can also
get two inequalities, but one can easily see that the two are, in fact, the same, after an appropriate 
permutation of the variables.

To illustrate the possibility to violate the inequalities \eqref{eq:main2}, let us take 
$\hat{A} = \hat{s}^a_x$, $\hat{A}' = \hat{s}^a_y$, 
$\hat{B} = \hat{s}^b_x$, $\hat{B}' = \hat{s}^b_y$ and apply these inequalities to the simplest Bell state 
\begin{equation}
    |B_2\rangle = \frac{1}{\sqrt{2}}(|00\rangle + |11\rangle).
\end{equation}
The inequality \eqref{eq:main2} for $n=2$ reads as
\begin{equation}\label{eq:main2-2}
\begin{split}
    &\langle\hat{s}^a_x\hat{s}^b_x + \hat{s}^a_x\hat{s}^b_y + \hat{s}^a_y\hat{s}^b_x\rangle^2 + 
    \langle\hat{s}^a_x\hat{s}^b_y + \hat{s}^a_y\hat{s}^b_x + \hat{s}^a_y\hat{s}^b_y\rangle^2 \\
    &+\langle\hat{s}^a_x\hat{s}^b_x - \hat{s}^a_y\hat{s}^b_y\rangle^2 \leqslant 8 - 6\langle\hat{s}^a_z\hat{s}^b_z\rangle,
\end{split}
\end{equation}
and for $n=4$ as
\begin{equation}\label{eq:main2-4}
\begin{split}
    &\langle\hat{s}^a_x\hat{s}^b_x + \hat{s}^a_x\hat{s}^b_y + \hat{s}^a_y\hat{s}^b_x\rangle^4 + 
    \langle\hat{s}^a_x\hat{s}^b_y + \hat{s}^a_y\hat{s}^b_x + \hat{s}^a_y\hat{s}^b_y\rangle^4 \\
    &+\langle\hat{s}^a_x\hat{s}^b_x - \hat{s}^a_y\hat{s}^b_y\rangle^4 \leqslant 34 - 32\langle\hat{s}^a_z\hat{s}^b_z\rangle.
\end{split}
\end{equation}
It is easy to check that $\langle\hat{s}^a_x\hat{s}^b_x\rangle = -\langle\hat{s}^a_y\hat{s}^b_y\rangle = 
\langle\hat{s}^a_z\hat{s}^b_z\rangle = 1$, $\langle\hat{s}^a_x\hat{s}^b_y\rangle = 
\langle\hat{s}^a_y\hat{s}^b_x\rangle = 0$. For the inequality \eqref{eq:main2-2} we have
$6 = 1^2 + (-1)^2 + 2^2 \leqslant 2$, which is clearly wrong. The inequality \eqref{eq:main2-4}
becomes $18 = 1^4 + (-1)^4 + 2^4 \leqslant 2$, which is also wrong, so both these inequalities can be violated.

In conclusion, we have developed a general approach to entanglement characterization 
and obtained simple and powerful entanglement conditions, which can
be used for discrete as well as for continuous-variable systems. Examples of
violations of these conditions are presented. Multipartite generalization of
some of our results are also discussed.

\end{document}